\documentclass [12 points] {article}
\usepackage{amssymb}
\usepackage{eucal}
\usepackage{amsmath}
\usepackage{amsthm}

\newcommand{\R}{{\mathbb  R}}

\newcommand{\ds}{\displaystyle}

\numberwithin{equation}{section}
\newtheorem{thm}{\bf Theorem}[section]

\newtheorem{prop}[thm]{\bf Proposition}

\newtheorem{defn}{\bf Definition}[section]
\newtheorem{corollary}[thm]{Corollary}
\theoremstyle{remark}
\newtheorem{rem}{\bf Remark}[section]

\newcommand{\de}{\delta}
\newcommand{\va}{\varphi}

\newcommand{\pa}{\partial}
\newcommand{\la}{\lambda}

\begin{document}

\title{\large\bf ASYMPTOTIC STABILITY FOR A CLASS OF METRIPLECTIC SYSTEMS }
\author{Petre Birtea, Mihai Boleantu, Mircea Puta, Razvan Micu Tudoran}
\date{ }
\maketitle

\begin{abstract}
Using the framework of metriplectic systems on $\R^n$ we will
describe a constructive geometric method to add a dissipation term
to a Hamilton-Poisson system such that any solution starting in a
neighborhood of a nonlinear stable equilibrium converges towards a
certain invariant set. The dissipation term depends only on the
Hamiltonian function and the Casimir functions.
\end{abstract}



\section{Introduction}

In an attempt for an unification of the conservative and
nonconservative (or dissipative) dynamics, A.N. Kaufman \cite{4}
has introduced the notion of metriplectic system. More exactly,
let $\{.,.\}$ be a Poisson structure on $\R^n$ with $\Pi$ the
associated matrix i.e., $\Pi=[\{x_i, x_j\}]$ and $C_1,..., C_k \in
C^\infty (\R^n, \R)$ a complete set of functionally independent
Casimir functions. Let $G$ be a smooth function from $\R^n$ to the
vector space of symmetric matrices of type $n \times n$.

\begin{defn}\label{d11}
(\cite{4}) A metriplectic system on $\R^n$ is a system of
differential equations of the following type:
\begin{equation}\label{e11}
    \dot x = \Pi(x) \cdot \nabla H(x) + G(x)\cdot \nabla \va (C_1,
    \dots, C_k)(x)
\end{equation} where $H \in C^\infty (\R^n, \R)$ and
$\va \in C^\infty (\R^k, \R)$ such that the following conditions
hold:
\begin{itemize}
  \item [(M1)]  $\Pi \cdot \nabla C_i =0, i=\overline{1,k}$, i.e. $C_i$ is
  a Casimir of our Poisson configuration $(\R^n, \{\cdot, \cdot\})$.
  \item [(M2)] $G \cdot \nabla H=0$.
  \item [(M3)] $(\nabla \va (C_1,
    \dots, C_k))^t \cdot G \cdot \nabla \va (C_1,\dots, C_k) \leq 0$.
\end{itemize}
\end{defn}

\begin{rem}\label{r11}
It is not hard to see that:
\begin{itemize}
  \item [(i)] The derivative of $H$ along the solutions of \eqref{e11} is
  $\ds\frac{dH}{dt}=0$, i.e. $H$ is a conserved  quantity
  of the dynamics \eqref{e11}.
  \item [(ii)]  The derivative of $\va (C_1,\dots, C_k)$ along the
  solutions of \eqref{e11} is
  $\ds\frac{d\va (C_1,\dots, C_k)}{dt} \leq 0$, i.e. $\va (C_1,\dots, C_k)$
  plays the role of the "entropy function" for the dynamics \eqref{e11}.

  \hfill{$\Box$}
\end{itemize}
\end{rem}

\begin{rem}\label{r12}
The dynamics \eqref{e11} can be viewed as a perturbation of the
Hamilton-Poisson system:
\begin{equation*}
    \dot x = \{x,H\},
\end{equation*} with the dissipative term $G \cdot \nabla
\va (C_1,\dots, C_k)$. \hfill{$\Box$}
\end{rem}

Metriplectic systems have been extensively studied in connection
with mathematical physics problems, see for instance \cite{1},
\cite{4}, \cite{5}, \cite{6} and \cite{7}. In \cite{bmr} was proven
that a certain type of dissipation induces instability. The goal of
our paper is to introduces a new type of dissipation, using the
formalism of metriplectic systems, in such a way that any solution
of (\ref{e11}) starting in a small neighborhood of a nonlinear
stable equilibrium converges towards a certain invariant set
containing the equilibrium. The dissipative part can be interpreted
as a set of controls added to the conservative part.

\bigskip

\section{A class of metriplectic systems on $\R^n$}

Let $(\R^n, \{\cdot, \cdot\}, H)$ be an Hamilton-Poisson system on
$\R^n$ and $C_1,...,C_k \in C^\infty (\R^n, \R)$ a complete set of
functionally independent Casimir functions of the Poisson vector
space $(\R^n, \{\cdot, \cdot\})$. The dynamic is described by the
following set of differential equations:
\begin{equation}\label{e21}
    \dot x = \Pi \cdot \nabla H.
\end{equation}

Our  goal is to construct  explicitly a dissipative perturbation,
i.e., to determine effectively a matrix $g = [g^{ij}]$ such the
perturbed system:
\begin{equation*}
    \dot x = \Pi \cdot \nabla H + G \cdot\nabla \va (C_1,\dots, C_k)
\end{equation*} to be a metriplectic system.

Denote with $\pa_i H \stackrel{def}{=} \ds\frac{\pa H}{\pa x_i}$ and
$\pa_i C_j \stackrel{def}{=} \ds\frac{\pa C_j}{\pa x_i}$. The matrix
$G=[g^{ij}]$ given below satisfies all the conditions from the
definition of a metriplectic system
\begin{equation}\label{e23}
    G=\left[\begin{array}{cccc} - \sum\limits_{\tiny \begin{array}{l}
    i=1\\i\not = 1\end{array}}^n (\pa_i H)^2& \pa_1 H \pa_2 H& \dots &
    \pa_1 H \pa_n H\\\pa_1 H \pa_2 H& - \sum\limits_{\tiny
    \begin{array}{l} i=1\\i\not = 2\end{array}}^n
    (\pa_i H)^2 & \dots & \pa_2 H \pa_n H\\ \vdots&
    \vdots&\dots&\vdots\\ \pa_1 H \pa_n H&
    \pa_2 H \pa_n H& \dots& - \sum\limits_{\tiny
    \begin{array}{l} i=1\\i\not = n\end{array}}^n(\pa_i H)^2
    \end{array}\right]
\end{equation}

Indeed, we have that the $j$-component of the vector field $G\cdot
\nabla H$ is given by

\begin{eqnarray*}
(G\cdot\nabla H)_j &=& \sum\limits_{\tiny i\neq j}^n
(\pa_i H)^2\pa_j H+(-\sum\limits_{\tiny i\neq j}^n (\pa_i H)^2)\pa_jH\\
&=&0.
\end{eqnarray*}
Consequently, condition $(M2)$ is satisfied.

In the case when in the dissipation term we take only one Casimir
function we have the following computation,

\begin{eqnarray*}
   (\nabla C)^t \cdot G\cdot \nabla C&=& \sum\limits_{j=1}^n {\pa_j
   C}(G\cdot\nabla C)_j=\sum\limits_{j=1}^n \pa_j
   (\sum\limits_{i\neq j}^n (\pa_i H \pa_j H \pa_i
   C-{(\pa_i H)}^2 \pa_j C))\\&=& \pa_1 C \pa_2 H (\pa_2C \pa_1 H - \pa_1 C \pa_2 H)
   + \dots + \pa_1 C \pa_n H (\pa_n C \pa_1 H - \pa_1 C \pa_n H) \\
   &+& \pa_2 C \pa_1 H (\pa_1 C \pa_2 H - \pa_2 C \pa_1 H) + \dots + \pa_2 C \pa_n H
   (\pa_n C \pa_2 H - \pa_2 C \pa_n H) \\&+&  \vdots\\
   &+& \pa_n C \pa_1 H (\pa_1 C \pa_n H - \pa_n C \pa_1 H) + \dots + \pa_n C \pa_{n-1}H
   (\pa_{n-1}C \pa_n H - \pa_n C \pa_{n-1}H).
\end{eqnarray*}
Regrouping the terms we obtain the desired inequality which proves
that condition $(M3)$ is also satisfied,
\begin{equation}\label{e22}
    \begin{array}{lllll} (\nabla C)^t \cdot G\cdot \nabla C & = - (\pa_1C \pa_2H -
    \pa_2C \pa_1H)^2&-(\pa_1C \pa_3H - \pa_3C \pa_1H)^2& - \dots & -
    (\pa_1C \pa_nH - \pa_nC \pa_1H)^2\\ &&-(\pa_2C \pa_3H-\pa_3C \pa_2H)^2&-\dots& -
    (\pa_2C \pa_nH - \pa_nC \pa_2H)^2\\&&& \ddots \\&&&& - (\pa_{n-1}C \pa_nH
    - \pa_nC \pa_{n-1}H)^2\\&\leq 0.
    \end{array}
\end{equation}

\begin{rem}\label{rem}
The above inequality is an equality iff $dH\wedge dC=0$.
Consequently, we obtain equality in (\ref{e22}) iff $\nabla H$ and
$\nabla C$ are linearly dependent. If, instead of one Casimir
function, we take a combination $\widetilde{C}:=\va (C_1,\dots,
C_k)$ of the complete set of functionally independent Casimir
functions $C_1,\dots, C_k$ we obtain a new Casimir function
$\widetilde{C}$ and consequently,
\begin{equation}\label{lyap}
(\nabla \va (C_1,\dots, C_k))^t \cdot G \cdot \nabla \va
(C_1,\dots, C_k) =(\nabla \widetilde{C})^t \cdot G \cdot \nabla
\widetilde{C} \leq 0
\end{equation}
with equality iff $\nabla H$ and $\nabla\va (C_1,\dots, C_k)$ are
linearly dependent.\hfill{$\Box$}
\end{rem}

\medskip

Let us consider now the metriplectic system \eqref{e11} where  $G$
is given by the relation \eqref{e23}. Then we can  define in a
canonical way two  vector fields on $\R^n$, namely:
\begin{equation*}
    \xi_{\Pi} = \Pi \cdot \nabla H
\end{equation*} and
\begin{equation*}
    \xi = \Pi \cdot \nabla H + G \cdot \nabla\va (C_1,\dots, C_k).
\end{equation*}

\begin{prop}\label{p21}
Let $(\R^n, \Pi, H)$ be an Hamilton-Poisson system. If $x_0 \in
\R^n$ is an equilibrium state of the vector field $\xi$, i.e.,
$\xi (x_0) =0$ then $x_0$ is an equilibrium state  of the vector
field $\xi_{\Pi}$.
\end{prop}

\begin{proof}
Indeed, $\xi (x_0)=0$ implies that
\begin{equation*}
    (\nabla\va (C_1,\dots, C_k)(x_0))^t \xi (x_0) =0,
\end{equation*} and so
\begin{equation*}
    (\nabla\va (C_1,\dots, C_k) (x_0))^t \Pi (x_0) \nabla H(x_0) + (\nabla\va (C_1,\dots, C_k)(x_0))^t
    G(x_0) \nabla\va (C_1,\dots, C_k)(x_0)=0.
\end{equation*} This is equivalent (since $\nabla\va (C_1,\dots, C_k)$ is a Casimir we have
$(\nabla\va (C_1,\dots, C_k) (x_0))^t \Pi (x_0) \nabla H(x_0)=0$)
with
\begin{equation*}
    (\nabla\va (C_1,\dots, C_k)(x_0))^t G(x_0) \nabla\va (C_1,\dots, C_k)(x_0) =0.
\end{equation*} This leads us immediately via Remark \ref{rem} to
\begin{equation*}
    \nabla H (x_0) = \la \nabla\va (C_1,\dots, C_k)(x_0)
\end{equation*} for some $\la \in \R^n$. Therefore
\begin{eqnarray*}
  \xi_{\Pi}(x_0) &=& \Pi (x_0) \nabla H (x_0)  \\
   &=&  \la \Pi (x_0)\nabla\va (C_1,\dots, C_k)(x_0)\\
   &=& 0
\end{eqnarray*}
as required.
\end{proof}

\begin{prop}\label{p22}
Let $(\R^n, \Pi, H)$ be an Hamilton-Poisson system. Let $x_0 \in
\R^n$ be an equilibrium point of the vector field $\xi_{\Pi}$. If
there exists a function $\va \in C^\infty (\R^k, \R)$ such that
$\nabla \va (C_1,\dots, C_k)(x_0)$ and $\nabla H (x_0)$ are linear
dependent then $x_0$ is an equilibrium point of the vector field
\begin{equation*}
\xi=\Pi \cdot \nabla H + G \nabla \va (C_1,\dots, C_k).
\end{equation*}
\end{prop}

\begin{proof}
Indeed, if
\begin{equation*}
    \xi_{\Pi} (x_0) =0,
\end{equation*} then we have also that
\begin{equation*}
    \Pi (x_0) \cdot \nabla H (x_0) =0
\end{equation*} and consequently we have two possibilities:

(i) $\nabla H (x_0) =0$. This implies (see the construction of
$G$) the equality
\begin{equation*}
    G(x_0)=0
\end{equation*}
and then
\begin{equation*}
    G(x_0) \cdot \nabla \va(C_1,\dots, C_k)(x_0)=0.
\end{equation*} Therefore
\begin{equation*}
    \xi (x_0)=0.
\end{equation*}

\medskip

(ii) $\Pi (x_0) \cdot \nabla H (x_0) =0$ and $\nabla H (x_0) \not
= 0$.  By hypothesis we obtain
\begin{eqnarray*}
  G (x_0) \cdot \nabla \va (C_1,\dots, C_k) (x_0)  &=& \lambda G(x_0) \nabla H (x_0)  \\
   &=& 0
\end{eqnarray*}
and we can conclude that
\begin{equation*}
    \xi (x_0)=0
\end{equation*} as required.
\end{proof}

\begin{corollary}\label{ld}The set $E:=\{x\in \R^n|\nabla H(x)$
and $ \nabla\va (C_1,\dots, C_k) (x)\ ${are linearly
dependent}$\}$ is a set of equilibrium points for both vector
fields $\xi_{\Pi}$ and $\xi$.
\end{corollary}

\begin{proof}For an arbitrary point $y\in E$ we have
\begin{equation*}
\Pi(y)\cdot \nabla H(y)=\lambda \Pi(y)\cdot \nabla \va (C_1,\dots,
C_k)(y)=0,
\end{equation*}
which shows that $y$ is a equilibrium point for the vector field
$\xi_{\Pi}$. Proposition \ref{p22} implies that $y$ is also an
equilibrium point for the vector field $\xi$.
\end{proof}

\bigskip

\section{Metriplectic systems and asymptotic stability}

For the beginning let us briefly recall some definitions of
stability for a dynamical system on $\R^n$
\begin{equation}\label{e31}
    \dot x = f(x),
\end{equation} where $f \in C^\infty (\R^n, \R^n)$.

\begin{defn}\label{d31}
An equilibrium point $x_e$ is stable if for any small neighborhood
$U$ of $x_e$ there is a neighborhood $V$ of $x_e$, $V \subset U$
such that if initially $x(0)$ is in $V$, then $x(t) \in U$ for all
$t
>0$. If in addition we have:
\begin{equation*}
    \lim_{t \to \infty} x(t) = x_e
\end{equation*}
then $x_e$ is called asymptotically stable.
\end{defn}

For studying more complicated asymptotic behavior we need to
introduce the notion of $\omega$-limit set. Let $\phi^t$ be the flow
defined by equation (\ref{e31}). The $\omega$-limit set of $x$ is
$\omega(x):=\{y\in \R^n|\exists t_1, t_2...\rightarrow \infty $
{s.t.} $ \phi(t_k, x)\rightarrow y $ {as} $ k\rightarrow\infty\}$.
The $\omega$-limit sets have the following properties that we will
use later. For more details see Robinson \cite{robinson}.

\begin{itemize}
  \item [(i)] If $\phi(t,x)=y$ for some $t\in \R$, then
  $\omega(x)=\omega(y)$.
  \item [(ii)] $\omega(x)$ is a closed subset and both positively and negatively
  invariant (contain complete orbits).
\end{itemize}

Next we will prove a version of LaSalle's Theorem. For a more
general result see the original work of LaSalle \cite{lasalle}.

\begin{thm}\label{las}Let $x_0$ be an equilibrium point for (\ref{e31})
and $U$ a small compact neighborhood of $x_0$. Suppose there
exists $L:U\rightarrow \R$ a $C^1$ differentiable function with
$L(x)>0, L(x_0)=0$ and $\dot L(x)\leq 0$. Let $E$ be the set of
all points in $U$ where $\dot L(x)=0$. Let $M$ be the largest
invariant set in $E$. Then there exists a small neighborhood
$V\subset U$ such that $\omega(x)\subset M$ for all $x\in V$.
\end{thm}

\begin{proof}The conditions of the theorem ensures the stability
of $x_0$. There is a small compact neighborhood $U$ of $x_0$ and a
smaller neighborhood $V\subset U$ such that $\phi(t,x)\in U$ for any
$x\in V$ and $t\geq 0$. As $U$ is closed we also have
$\omega(x)\subset U$.

Let $x\in V$ and $y\in \omega(x)$ be arbitrarily chosen. Since $\dot
L(\phi(t,x))\leq 0$ we have that $L(\phi(t,x))$ is a nonincreasing
function of $t$. Because $L$ is a positive bounded function on $U$
and $\phi(t,x)$ remains for all time in $U$ we have
$\lim_{t\rightarrow \infty} L(\phi(t,x))=l$, where $0\leq l<
\infty$. As $y\in \omega(x)$ and $L$ is continuous we obtain that
$L(y)=l$. The invariance of $\omega(x)$ shows that $L(\phi(t,y))=l$
and consequently $\dot L(\phi(t,y))=0$ for all $t\in \R$. Hence
$y\in E$ and so $\omega(x)\subset E$. As $\omega(x)$ is invariant
implies that $\omega(x)\subset M$.
\end{proof}

The following is the main result of this paper.

\begin{thm}\label{p31}
Let $(\R^n, \Pi, H)$ be a Hamilton-Poisson system and $x_0 \in
\R^n$ an equilibrium state for the dynamic
\begin{equation}\label{e32}
    \dot x = \Pi(x) \cdot \nabla H (x).
\end{equation} Suppose that there exists a function $\va \in C^\infty (\R^k,
\R)$, where $k$ equals the number of functionally independent
Casimirs for the Poisson structure $\Pi$, such that
\begin{itemize}
  \item [(i)] $\de H_\va (x_e)=0$
  \item [(ii)] $\de^2 H_\va(x_e)$ is positive definite,
\end{itemize}
where
  \begin{equation*}
    H_\va (x) = H(x) + \va (C_1 (x), \dots , C_k (x))
  \end{equation*} with  $C_1, \dots C_k \in C^\infty (\R^n, \R)$ a set
  of functionally independent Casimirs of $\Pi$.

  Let $G$ be the matrix defined by \eqref{e23} then there exist a small closed and bounded
  neighborhood $U$ of the equilibrium state $x_e$ of the corresponding metriplectic system

\begin{equation}\label{e33}
    \dot x = \Pi(x) \cdot \nabla H(x) + G(x) \cdot  \nabla \va (C_1,
    \dots, C_k) (x)
\end{equation}
and a neighborhood $V\subset U$ such that every solution of
\eqref{e33} starting in $V$ approaches $U\cap E$ as
$t\rightarrow\infty$, where $E:=\{x\in \R^n|\nabla H(x)$ and $
\nabla\va (C_1,\dots, C_k) (x)\ ${are linearly dependent}$\}$.
\end{thm}

\begin{proof}
First we have to prove that $x_e$ is an equilibrium point for the
dynamics \eqref{e33}. This is guarantied by Proposition \ref{p22}.

Next we will prove that under the hypothesis of the theorem the
function $L \in C^\infty (\R^n, \R)$ given by
\begin{equation*}
    L(x) \stackrel{def}{=}  H_\va (x) - H_\va (x_e)
\end{equation*} is a Lyapunov function for the dynamic \eqref{e33}.
More exactly, using the hypothesis and Remark \ref{rem}, we obtain
that there exists a closed and bounded neighborhood $U$ of the
critical point $x_e$ such that
\begin{itemize}
  \item [(i)] $L(x_e) =0$.
  \item [(ii)] $L(x) >0$, $(\forall) \  x \in U$, $ x \not =
  x_e$
  \item [(iii)] $\dot L (x) \leq 0 $, $(\forall) \  x\in
  U$,
\end{itemize} which implies that $x_e$ is a stable equilibrium
point for \eqref{e33}.

By Remark \ref{rem} we have that $E$ equals the set of all points
where $\dot L(x)=0$. Using Corollary \ref{ld} we have that the
largest invariant subset in $E$ for \eqref{e33} equals $E$.

We showed that all the conditions of the Theorem \ref{las} are
satisfied and so we obtain the desired result.

\end{proof}

\begin{rem}\label{r31}
The above result tells us in fact how to built in an effective way
a set controls which locally asymptotically stabilize a nonlinear
stable equilibrium state of a given Hamilton-Poisson dynamics.

\hfill{$\Box$}
\end{rem}

\bigskip

\section{Examples}

It is well known that Euler's angular momentum equations of the
free rigid body can be written on $\mathbb{R}^3$ in the following
form:
\begin{equation}\label{rigid}
\left\{ \begin{array}{l}
 \mathop {x_1 }\limits^ \cdot   = a_1 x_2 x_3  \\
 \mathop {x_2 }\limits^ \cdot   = a_2 x_1 x_3  \\
 \mathop {x_3 }\limits^ \cdot   = a_3 x_1 x_2  \\
 \end{array} \right.
\end{equation}
where
$$a_1  = \dfrac{1}{{I_3 }} - \dfrac{1}{{I_2 }};\,\,\,a_2  =
\dfrac{1}{{I_1 }} - \dfrac{1}{{I_3 }};\,\,\,\,a_3  =
\dfrac{1}{{I_2 }} - \dfrac{1}{{I_1 }};$$\\
$I_1,I_2,I_3$ being the components of the inertia tensor and we
suppose as usually that
$$I_1>I_2>I_3>0 .$$

The equations \eqref{rigid} have the following Hamilton-Poisson
realization:
$$((so\,(3))^*  \approx \mathbb{R}^3,\{  \cdot , \cdot \} \_\,,H)$$\\
where $\{  \cdot , \cdot \} \_$ is minus-Lie-Poisson structure on
$((so\,(3))^*  \approx \mathbb{R}^3$, generated by the matrix:
\begin{equation*}
\sqcap \_ = \left[ \begin{array}{l}
 \,\,\,\,0\,\,\,\,\,\,\, - x_3 \,\,\,\,\,\,\,\,x_2  \\
 \,\,\,x_3 \,\,\,\,\,\,\,\,\,0\,\,\,\,\,\,\, - x_1  \\
  - x_2 \,\,\,\,\,\,\,\,x_1 \,\,\,\,\,\,\,\,\,0 \\
 \end{array} \right]
\end{equation*}
and the Hamiltonian $H$ is given by:
\begin{equation}
H(x_1 ,x_2 ,x_3 ) = \dfrac{1}{2}\left( {\dfrac{{x_1^2 }}{{I_1 }} +
\dfrac{{x_2^2 }}{{I_2 }} + \dfrac{{x_3^2 }}{{I_3 }}} \right).
\end{equation}

It is not hard to see that the function $C \in C^
\infty(\mathbb{R}^3, \mathbb{R})$ given by:
\begin{equation}
C(x_1 ,x_2 ,x_3 ) = \dfrac{1}{2}(x_1^2+ x_2^2+ x_3^2 )
\end{equation}
is a Casimir of our configuration $((so\,(3))^*  \approx
\mathbb{R}^3,\{  \cdot , \cdot \} \_)\,.$

Let $(M_0, 0, 0)$ be an equilibrium point for \eqref{rigid}. The
function $H_\va (x) = H(x) + \va (C(x))$, where $\varphi$ is given
for instance by
$$\varphi (s) = \left( {s - \dfrac{1}{2}M_{0}^2 } \right)^2  - \dfrac{s}{{I_1 }}\,\,,$$
satisfies the conditions $(i)$ and $(ii)$ of Theorem \ref{p31}. In
this case the perturbed system \eqref{e33} is given by

\begin{equation}\label{perturbat}
\left\{ \begin{array}{l}
 \mathop {x_1 }\limits^ \cdot   = a_1 x_2 x_3 +x_1(x_1^2+ x_2^2+ x_3^2 -c)
 (-\dfrac{a_3}{{I_2 }}x_2^2+\dfrac{a_2}{{I_3 }}x_3^2) \\
 \mathop {x_2 }\limits^ \cdot   = a_2 x_1 x_3 +x_2(x_1^2+ x_2^2+ x_3^2 -c)
 (\dfrac{a_3}{{I_1 }}x_1^2-\dfrac{a_1}{{I_3 }}x_3^2)\\
 \mathop {x_3 }\limits^ \cdot   = a_3 x_1 x_2 + x_3(x_1^2+ x_2^2+ x_3^2 -c)
 (-\dfrac{a_2}{{I_1 }}x_1^2+\dfrac{a_1}{{I_2 }}x_2^2)\\
 \end{array} \right.
\end{equation}
where $c=M_0^2-\dfrac{1}{{I_1}}$.

The set of points in $\R^3$ where $\nabla H(x)$ and
$\nabla\va(C(x))$ are linearly dependent is given by
$$E=\{(\lambda, 0, 0)|\lambda\in \R\}\bigcup \{(0, \lambda, 0)|\lambda\in \R\}
\bigcup \{(0, 0, \lambda)|\lambda\in \R\}.$$

By Theorem \ref{p31} we obtain that there exists a small closed and
bounded neighborhood $U \subset \R^3$ around the equilibrium point
$(M_0, 0, 0)$ and $V\subset U$ such that any solution of
\eqref{perturbat} starting in $V$ approaches the set
$\{(M_0+\lambda, 0, 0)|\lambda\in [-\varepsilon,\varepsilon]\subset
\R \}$ as $t\rightarrow\infty$.

A similar result with obvious modifications can be also obtain for
the equilibrium state $(0,0,M_0)$, $M_0\in \R$.

\bigskip

{\bf Acknowledgments.} Petre Birtea and Mircea Puta were partially
supported by the program SCOPES and the Grant CNCSIS 2007/2008 and
Razvan Micu Tudoran was partially supported by the program SCOPES
and the Grant CNCSIS AT 60/2007.

\end{document}